\documentclass[11pt,a4paper]{article}
\usepackage[totalwidth=13.0cm,totalheight=20.0cm]{geometry}
\usepackage{latexsym,amsthm,amsmath,amssymb,url}
\usepackage[ruled, linesnumbered]{algorithm2e}
\usepackage[all,british]{foreign} 

\usepackage{tikz,authblk}
\usepackage{skak}
\usetikzlibrary{decorations.pathreplacing}

\newtheorem{theorem}{Theorem}

\begin{document}
\title{Searching for Maximum Out-Degree \\
Vertices in Tournaments}

\date{}

\author[1]{Gregory Gutin\thanks{Partially supported by the Royal Society Wolfson Research Merit Award.}}
\author[2]{George B. Mertzios\thanks{Partially supported by the EPSRC Grant~EP/P020372/1.}}
\author[1]{Felix Reidl}

\affil[1]{Dept.~of Comp. Science, Royal Holloway, University of London, UK. \newline
Email: \texttt{g.gutin@rhul.ac.uk}, \texttt{felix.reidl@rhul.ac.uk}}
\affil[2]{Department of Computer Science, Durham University, UK. \newline
Email: \texttt{george.mertzios@durham.ac.uk}}



\maketitle

\pagestyle{plain}

\begin{abstract}
A vertex $x$ in a tournament $T$ is called a king  if for every vertex $y$ of
$T$ there is a directed path from $x$ to $y$ of length at most 2. It is not
hard to show that every vertex of maximum out-degree in a tournament is a
king. However, tournaments may have kings which are not vertices of maximum
out-degree. A binary inquiry asks for the orientation of the edge between a
pair of vertices and receives the answer. The cost of finding a king in an
unknown tournament is the number of binary inquiries required to detect a
king. For the cost of finding a king in a tournament, in the worst case, Shen,
Sheng and Wu (SIAM J. Comput., 2003) proved a lower and upper bounds  of
$\Omega(n^{4/3})$ and $O(n^{3/2})$, respectively. In contrast to their result,
we prove that the cost of finding a vertex of maximum out-degree is ${n
\choose 2} -O(n)$ in the worst case.\\

\end{abstract}

\section{Introduction}

A {\em tournament} is an orientation of a complete graph. A vertex $x$ in a
tournament $T$ is called a {\em king}  if for every vertex $y$ of $T$ there is
a directed path from $x$ to $y$ of length at most 2. We use standard
terminology and notation for directed graphs, see e.g. \cite{BangG09}. 

Kings in tournaments were introduced by Landau \cite{Landau53} who studied
animal societies and observed that in every animal society ``there are members
who dominate every other member either directly or indirectly through a single
intermediate member.'' Inspired by this observation, Landau \cite{Landau53}
proved that every tournament has a king. His result can easily be shown as
follows. Let $x$ be a vertex of maximum out-degree (MOD) in a tournament
$T=(V,A)$ and let $y$ be another vertex of $T$. Then either $xy\in A$ or there
is a vertex $z$ such that both $xz$ and $zy$ are arcs of $T$ (otherwise, $x$
is not of maximum out-degree). However, it is not hard to construct
tournaments where not only MOD vertices are kings.

An ordering $x_1,x_2,\dots ,x_n$ of vertices of a tournament $T=(V,A)$ is
called a {\em sorted sequence of kings} if $x_ix_{i+1}\in A$ for $i=1,2,\dots
,n-1$ and $x_j$ is a king in the subgraph of $T$ induced by the vertices
$x_j,x_{j+1},\dots ,x_n.$ Lou, Wu and Sheng \cite{LouWS00} showed that every
tournament has a sorted sequence of kings. Sorted sequence of kings are of
interest due to their useful properties, e.g. every such sequence is a median
order of $T$ \cite{ShenSW03}, i.e. an ordering of $V$ such that the number of
arcs $x_jx_i$ with $j>i$ is minimum possible. Median orders in tournaments
were used, among other things, by Havet and Thomass\'e to prove Seymour's
Second Neighborhood Conjecture restricted to tournaments \cite{HavetT00}.

Wu and Sheng \cite{WuS01} designed an $O(n^2)$-time algorithm for finding a
sorted sequence of kings in a tournament with $n$ vertices. A {\em binary
inquiry} asks for the orientation of the edge between a pair of vertices and
receives the answer. Shen, Sheng and Wu \cite{ShenSW03} studied the {\em cost}
of finding a sorted sequence of kings, which is the number of binary inquiries
required to find such a sequence (including a king) in an unknown tournament.
They proved that the cost is $\Theta(n^{3/2})$ in the worst case. Shen, Sheng
and Wu \cite{ShenSW03} also studied the cost of finding a king in a
tournament. They showed a lower bound $\Omega(n^{4/3})$ for the cost in the
worst case. Clearly, $O(n^{3/2})$ is an upper bound. This has left a cost gap
between $\Omega(n^{4/3})$ and $O(n^{3/2})$ for finding a king, and closing
this gap remains an open problem \cite{ShenSW03}.

In this short paper, we prove that the cost of finding an MOD vertex is ${n
\choose 2} - O(n)$ in the worst case. Thus, the worst case costs of finding an
arbitrary king and that of maximum out-degree are quite different and while searching 
for kings, it is not a good idea to restrict ourselves to only MOD vertices.

As a warming-up, consider a related question of the cost of deciding whether a
tournament $T$ has a vertex of in-degree zero; clearly such a vertex is a king
and $T$ has no other kings.\footnote{It is well-known that a tournament
without vertices of in-degree zero has at least three kings \cite{BangG09}.}
Consider the following procedure. Repeat until $T$ has only one vertex: choose
a maximum matching $M$ in the complete graph with vertices $V(T)$, ask for the
orientation of each edge of $M$ in $T$ and for every $\{x,y\}\in M$ delete from $T$
vertex $x$ or $y$ if $xy$ is oriented towards $x$ or $y$, respectively.
After completion, check whether the last remaining
vertex $z$ has no arcs to it from the other vertices in the original $T$. Now
it is not hard to see that the cost of deciding whether $T$ has a vertex of
in-degree zero is less than $2n$ in the worst case, where $n$ is the number of
vertices in $T$. Indeed, the number of inquiries in the repeat loop is less
than $n$ and so is the number of inquiries for orientations of edges incident
to $z$.

\section{Finding an MOD Vertex}

A tournament $T$ on an odd number $n$ of vertices is {\em
regular} if  every vertex of $T$ is of out-degree of $(n-1)/2.$ Thus, all 
in-degrees in a regular tournament are also $(n-1)/2.$ A tournament $T$ on an 
even number $n$ of vertices is {\em almost regular} if  every vertex of $T$ is of
out-degree $n/2$ or $n/2 -1.$ Observe that an almost regular tournament has
$n/2$ vertices of out-degree $n/2$ and in-degree $n/2 -1$ and $n/2$ vertices
of out-degree $n/2-1$ and in-degree $n/2.$

\begin{theorem} 
  The worst case cost of finding an MOD vertex in a tournament
  on $n$ vertices is at least  $(n-1)^2/2$ when $n$ is odd and at least
  $(n-1)(n-2)/2$ when $n$ is even.
\end{theorem}
\begin{proof} 
Let $T=(V,A)$ be a tournament on $n$ vertices. Consider an algorithm which
determines an MOD vertex of $T$ using binary inquiries. The algorithm knows
$V$, but initially no element of $A$.  Let $q$ be the number of inquiries used
by the algorithm before it correctly returns an MOD vertex $x$ of $T$. Let us
denote by $d^+(y)$ and $d^-(y)$ the number of arcs known to be oriented from
and towards $y$, respectively, after the $q$ inquiries are used.

As in \cite{ShenSW03}, we introduce an adversary (to the algorithm) that
answers questions of the binary inquiries in such a way as to delay the time
the algorithm returns $x$. However, our strategy for the adversary is
different from that in  \cite{ShenSW03}. The
adversary has $T$ as a concrete regular or almost regular tournament
(depending on parity of $n$) and answers the inquiries according to the
orientations of the arcs in $T$. However, after the algorithm returns $x$, the
adversary checks whether the remaining ${n \choose 2} -q$ inquiries can be
answered in a such a way (not necessarily according to the  orientations of
arcs in $T$) that $x$ is not an MOD vertex. Let us assume that the algorithm
correctly decided that $x$ is an MOD vertex and consider two cases.

\begin{description}
\item[ Case 1:]   $n$ is odd. Suppose that $d^+(x)<(n-1)/2.$ Then the
adversary can answer all remaining questions on arcs incident with $x$ by
saying the arcs are all oriented towards $x$ and answer all other remaining
questions according to the arcs in $T$. Then in the resulting tournament $T'$,
$x$ will be of out-degree less than $(n-1)/2,$ but another vertex of $T'$ will
be of out-degree at least $(n-1)/2$, which is a contradiction to our
assumption that the algorithm is correct and $x$ is an MOD vertex. Thus,
$d^+(x)\ge (n-1)/2$ and since $T$ is regular we have $d^+(x) = (n-1)/2.$

Suppose there is a vertex $y\neq x$ such that $d^-(y)<(n-1)/2.$ Then the
adversary can answer all remaining questions on arcs incident with $y$ by
saying that they are all oriented away from $y$ and answer all other remaining
questions according to the arcs in $T$. Then the out-degree of $y$ in the
resulting tournament will be larger than $(n-1)/2,$ a contradiction. Thus, for
all $y\neq x$ we have $d^-(y)\ge (n-1)/2$. Therefore, the number of queries must
satisfy
\[ 
  q\ge \sum_{z\in V\setminus \{x\}} d^-(z)\ge (n-1)^2/2.
\]

\item[ Case 2:]   $n$ is even. Recall that half of the vertices in $T$ are of
out-degree $n/2$ and the others of out-degree $n/2 - 1.$ Similarly to Case 1,
we can show that $d^+(x)=n/2.$ Suppose that there is a vertex $y\neq x$ with
$d^-(y)\le n/2 - 2.$ Then the adversary can answer all remaining questions on
arcs incident with $y$ by saying that they are all oriented away from $y$ and
get $d^+(y)>n/2,$ a contradiction. Thus, for all $y\neq x$ we have $d^-(y)\ge
n/2 -1.$ Therefore, the number of queries must
satisfy
\[ 
  q\ge \sum_{z\in V\setminus \{x\}} d^-(z)\ge (n-1)(n-2)/2.
\]
\end{description}

\noindent
The two lower bounds on~$q$ derived above now prove the claim.

\end{proof}


\end{document}